\documentclass[12pt]{amsart}
\usepackage{amssymb, amsmath, hyperref}
\pdfoutput=1

\def\<{\langle}
\def\>{\rangle}
\def\Z{{\mathbb Z}}

\def\R{{\mathbb R}}

\def\cC{{\mathcal C}}
\def\B{{\mathcal B}}
\def\E{{\mathcal E}}
\def\D{{\mathcal D}}

\def\W{{\mathcal W}}
\def\J{{\mathcal J}}
\def\wD{{\widetilde {\mathcal D}}}
\def\C{\mathbb{C}}

\DeclareMathOperator{\tr}{tr}
\DeclareMathOperator{\diag}{diag}

\newtheorem{theorem}{Theorem} 
\newtheorem{lemma}[theorem]{Lemma} 
\newtheorem{corollary}[theorem]{Corollary}
\newtheorem{proposition}[theorem]{Proposition}
\theoremstyle{definition}
\newtheorem*{notation}{Notation}
\newtheorem{definition}[theorem]{Definition}
\theoremstyle{remark}
\newtheorem{example}[theorem]{Example}
\newtheorem{remark}[theorem]{Remark}

\begin{document}

\title[CP map from basis to dual basis]{Complete positivity of the  map from a basis to its dual basis}

\author[V.~I.~Paulsen]{Vern I.~Paulsen}
\address{Department of Mathematics, University of Houston,
Houston, TX 77204-3476}
\email{vern@math.uh.edu}
\author[Fred~Shultz]{Fred Shultz}
\address{Department of Mathematics, Wellesley College, Wellesley, MA 02481}
\email{fshultz@wellesley.edu}

\thanks{The research of the first author was supported by NSF grant 1101231}
%\keywords{frame}

\begin{abstract} The dual of a matrix ordered space has a natural matrix ordering that makes the dual space matrix ordered as well. The purpose of these notes is to give a condition that describes when the linear map taking a basis of $M_n$ to its dual basis is a complete order isomorphism.
We  exhibit ``natural'' orthonormal bases for $M_n$  such that this
map is an order isomorphism, but not a complete order
isomorphism. Included among such bases is the Pauli basis. Our results generalize the Choi matrix by giving conditions under which the role of the standard basis $\{E_{ij}\}$ can be replaced by other bases.
\end{abstract}

\maketitle
\date{\today}

Given a vector space $V$ there is no ``natural'' linear isomorphism between $V$ and the dual space $V^d,$ but each time we fix a basis $\B= \{v_i: i \in I \}$ for $V$ there is a {\em dual basis} $\widetilde{\B}= \{ \delta_i: i \in I\}$ for $V^d$ satisfying 
\[\delta_i(v_j) = \begin{cases} 0, & i \ne j\\ 1, & i=j \end{cases}\]
and this allows us to define a (basis dependent) linear isomorphism
between $V$ and $V^d.$

 \begin{definition} If $\B$ is a  basis of $M_n$, the linear map from $M_n$ to $M_n^d$ taking each member of $\B$ to the corresponding member of the dual basis is denoted by $\D_\B$, and is called the \emph{duality map}. We let $\Gamma_\B= \D_\B^{-1}: M_n^d \to M_n$ denote the inverse of this map.
%  and call it the {\em inverse duality map.  

\end{definition}
   
Note that if $f \in M_n^d,$ and $\B$ is a basis of $M_n$, then $\Gamma_\B(f) = \sum_{b\in \B} f(b) b$.
In particular, when $V= M_n$ (the space of $n \times n$ complex matrices), and we let $\E= \{ E_{i,j}: 1 \le i,j \le n \}$ denote the standard matrix units, then the map $\Gamma_\E: M_n^d \to M_n$ satisfies
\[ \Gamma_\E(f) = \sum_{i,j=1}^n f(E_{i,j}) E_{i,j}. \]
%and it is just the usual map that identifies a functional $f$ with its density matrix $\Gamma_\E(f).$ 
\begin{definition}
If $f \in M_n^d$, there is a unique matrix $D$ such that $f(X) = \tr(DX)$ for all $X \in M_n$, and we call this matrix the {\it density matrix} for $f$, with no requirement of positivity for $f$ or $D$.
\end{definition}

  Thus $\Gamma_\E$  is  just the map that identifies a functional $f$ with the transpose of its density matrix:
\begin{equation}\label{eq:density matrix}
f(X) = \tr (\Gamma_\E(f)^t X) \text{ for all $X \in M_n$.}
\end{equation}
%and it is just the usual map that identifies a functional $f$ with its density matrix $\Gamma_\E(f).$ 

This note is motivated by the following result of Paulsen-Todorov-Tomforde \cite[Thm. 6.2]{PTT}, which we will see later is a restatement of the Choi-Jamio\l kowski correspondences \cite{Choi,Jamiol}.
 
  \begin{theorem}\label{PTT} The map $\D_\E:M_n \to M_n^d$ is a complete order isomorphism between these matrix ordered spaces.  
  \end{theorem}

In this paper, we will show that this theorem is very dependent on the choice of basis. In fact, we will show that there exist orthonormal bases for $M_n$ such that the inverse duality map does not even send positive functionals to positive matrices. Even more intriguing, we will show that there are ``natural'' orthonormal bases for $M_n$ such that the inverse duality map does send positive functionals to positive matrices, yet does not send completely positive maps to positive block matrices. These results can be interpreted as giving some new Choi-Jamio\l kowski type results.

We believe that such bases might be useful as entanglement witnesses.  (We will comment further on this after Corollary \ref{genlchoi}).

Before proceeding it will be necessary to establish some notation. 
 Recall that when we say that a vector space $V$ is {\em matrix ordered} we mean that for each natural number $p,$ we have specified a cone $\cC_p$ in the vector space of $p \times p$ matrices over $V,$ $M_p(V),$ which we identify as the {\em positive elements} in $M_p(V)$, and that these cones must satisfy certain natural axioms, such as if $A \in \cC_p$ and $B \in \cC_q,$ then $A \oplus B \in \cC_{p+q}.$ We also require that if $X=(x_{i,j})$ is a $p \times q$ matrix of scalars and $A=(v_{i,j}) \in \cC_p,$ then
\[ X^*AX= (\sum_{k,l=1}^p \overline{x_{k,i}}v_{k,l} x_{l,j}) \in \cC_q.\]
When there is no ambiguity we simply write $\cC_p=M_p(V)^+.$ See \cite[Chpt. 13]{Paulsen} for more background on matrix ordered spaces.

Matrix ordered spaces are the natural setting for studying {\em completely positive maps.} Indeed, given matrix ordered spaces $V, W$ we say that a linear map $\Phi:V \to W$ is completely positive provided that for each $p,$ $(v_{i,j}) \in M_p(V)^+$ implies that $(\Phi(v_{i,j})) \in M_p(W)^+.$  (We will denote the map that takes $(v_{i,j})$  to $(\Phi(v_{i,j}))$ by $\Phi^{(p)}$, so that $\Phi$ is completely positive iff every map $\Phi^{(p)}$ is a positive map.) We say that $\Phi$ is a complete order isomorphism provided that $\Phi$ is invertible and that $\Phi$ and $\Phi^{-1}$ are both completely positive.

The most frequently encountered example of a matrix ordered space is $L(H),$ the bounded linear operators on a Hilbert space $H.$ We define the matrix ordering by identifying $M_p(L(H))= L(H \oplus \cdots \oplus H),$ the bounded linear operators on the direct sum of $p$ copies of the Hilbert space, and declaring $(A_{i,j}) \in M_p(L(H))^+$ exactly when it defines a positive operator on the Hilbert space $H \oplus \cdots \oplus H.$

 When $H= \C^n,$ the standard $n$-dimensional Hilbert space, we write
 $e_1, \ldots, e_n$ for the standard basis of $\C^n$, and write $\<
 \cdot\,, \cdot \>$ for the inner product on $\C^n$. We identify
 $L(\C^n)$ with $ M_n$, the set of $n \times n$ matrices with entries in
 $\C.$ Note that identifying $M_p(M_n),$ the $p \times p$ block
 matrices with entries from $M_n,$ with $M_{pn}$ yields the same cone
 of positive matrices as when we identify $M_p(M_n)$ with the linear
 maps on the direct sum of $p$ copies of $\C^n,$ $L(\C^n \oplus \cdots
 \oplus \C^n).$

Finally,
given a matrix ordered space $V$ there is a natural way to define a matrix   
ordering on the dual space $V^d.$ To do this we declare that a matrix of functionals $(f_{i,j}) \in M_p(V^d)$  belongs to $M_p(V^d)^+$ if and only if the linear map
$\Phi: V \to M_p$ given by $\Phi(v) = (f_{i,j}(v))$ is completely positive.

In this paper we will be concerned with examining various natural bases $\B$ for $M_n$ and determining whether or not the duality map $D_\B$ is a complete order isomorphism. We will see that there exist bases for $M_n$ such that $D_\B$ is positive but not completely positive. 

Since our results rely on Theorem~\ref{PTT}, we present a new proof here, that
is somewhat simpler than the proof that appeared in \cite{PTT}.

\medspace

{\em Proof of Theorem \ref{PTT}.} Rather than proving that $\D_\E$ is
a complete order isomorphism, we prove, equivalently, that $\Gamma_\E
= \D_\E^{-1}:M_n^d \to M_n$ is a complete order isomorphism. We have already seen
that $\Gamma_\E$ sends functionals to the transpose of their density matrices. Since a
functional is positive if and only if its density matrix is positive, and the transpose map is an order isomorphism,
we see that $\Gamma_\E$ is an order isomorphism.

Now let $(f_{k,l}) \in M_p(M_n^d)$ and consider the map $\Phi:
M_n \to M_p$ defined by $\Phi(X)= (f_{k,l}(X))= \sum_{k,l=1}^p f_{k,l}(X)E_{k,l}.$ We must show that
$\Phi$ is completely
positive if and only if $\Gamma_\E^{(p)}((f_{k,l})) \in M_p(M_n)^+.$

%We have that 
%\begin{multline*} \Gamma_\E^{(p)}((f_{k,l})) = \big(\Gamma_\E(f_{k,l})\big)_{k,l=1}^p = \big( \sum_{i,j=1}^n
%f_{k,l}(E_{i,j})E_{i,j}) \big)_{k,l=1}^p \\ = \sum_{k,l=1}^p [\sum_{i,j=1}^n
%f_{k,l}(E_{i,j})E_{i,j}] \otimes E_{k,l} = \sum_{i,j=1}^n E_{i,j}
%\otimes [\sum_{k,l=1}^p f_{k,l}(E_{i,j})E_{k,l}] \\ = \sum_{i,j=1}^n
%E_{i,j} \otimes \Phi(E_{i,j}). \end{multline*}
%But this last matrix is the Choi-Jamliakowska density matrix and by
%Choi's theorem\cite{Choi} the map $\Phi$ is completely positive if and
%only if this block matrix is positive.

We have  
\begin{align*} \Gamma_\E^{(p)}((f_{k,l})) &= \big(\Gamma_\E(f_{k,l})\big)_{k,l=1}^p = \big( \sum_{i,j=1}^n
f_{k,l}(E_{i,j})E_{i,j}) \big)_{k,l=1}^p \cr &= \sum_{k,l=1}^p E_{k,l}\otimes [\sum_{i,j=1}^n
f_{k,l}(E_{i,j})E_{i,j}]  
= \sum_{i,j=1}^n 
[\sum_{k,l=1}^p f_{k,l}(E_{i,j})E_{k,l}]\otimes E_{i,j} \cr &= \sum_{i,j=1}^n\Phi(E_{i,j})\otimes E_{i,j}. \end{align*}
Since the  map that takes $A \otimes B$ to $B \otimes A$ extends to a *-isomorphism of $M_p \otimes M_n$ onto $M_n \otimes M_p$, the last expression is positive iff the matrix 
\begin{equation}\label{Choi} C_\Phi =  \sum_{i,j=1}^nE_{i,j}\otimes \Phi(E_{i,j})
\end{equation}
is positive. 
But this last matrix is the Choi matrix and by
Choi's theorem~\cite{Choi} the map $\Phi$ is completely positive if and
only if this block matrix is positive.
Thus $(f_{k,l}) \in M_p(M_n)^+$ if and only if $(\Gamma_\E(f_{k,l}))
\in M_p(M_n)^+$ and we have shown that $\Gamma_\E$ is a complete order
isomorphism. This completes the proof of Theorem~1. \hfill$\square$

\medspace
\medspace

 A map $\Psi:M_n \to M_n$ is called a {\em co-positive order isomorphism} provided that its composition $t \circ \Psi$ with the transpose map $t$ on $M_n$ is a complete order isomorphism. 
 
\begin{corollary}\label{densitymatrixcor} The linear map from $M_n^d$ to $M_n$ that takes a functional to its density matrix is a co-positive order isomorphism.
\end{corollary} 

\begin{proof}
As remarked in connection with equation \eqref{eq:density matrix}, the map that takes a functional to its density matrix is $t \circ \Gamma_\E$.
\end{proof}

Now that we have a complete order isomorphism $\D_\E$ between $M_n$ and
$M_n^d,$ in order to determine whether or not other maps between $M_n$ and $M_n^d$
are complete order isomorphisms,
it will be convenient to work with a map $\wD_\B: M_n \to M_n$  instead of $\D_\B:M_n \to M_n^d$.

\begin{definition} Let $\B$ be a basis of $M_n$ and $\E$ the standard basis of matrix units.  Then we define $\wD_\B: M_n \to M_n$ by $\wD_\B = \Gamma_\E\circ \D_\B$.
\end{definition}

Note that since $\Gamma_\E$ is a complete order isomorphism, $\D_\B$
will be a complete order isomorphism if and only if $\wD_\B$ is a
complete order isomorphism.

Given a matrix $A \in M_n$ we let $A^t$ denote its transpose.  Recall that $\E = \{E_{ij}\}$ is an orthonormal basis for $M_n$ with respect to the Hilbert-Schmidt inner product, which we  denote by $\< A, B\>= \tr(AB^*)$, where $\tr:M_n \to \C$ denotes the unnormalized trace, $\tr(A) = \sum_{i+1}^n a_{i,i}$.

Using the  orthonormal basis $\E$ for $M_n$, we represent
elements in $L(M_n)$ as $n^2 \times n^2$ matrices. Given $L \in
L(M_n)$ we write $L^T$ to denote the transpose of the matrix for $L$
with respect to the basis $\E$, to distinguish this transpose from the
transpose on $M_n.$  

%\begin{definition} Let $\B$ be a basis of $M_n$ and $\E$ the standard basis of matrix units.  Then a \emph{change of basis map} $C_\B \in GL(M_n)$ denotes any linear map taking the set $\E$ to the set $\B$.
%\end{definition}

\begin{definition} Let $\B$ be a basis of $M_n$ and $\E$ the standard basis of matrix units. A  \emph{change of basis map}  is any linear map $C_\B$ in $ L(M_n)$ taking the set $\E$ to the set $\B$.  By slight abuse of notation, we write $C_\B^T$ for the unique linear map in $L(M_n)$ whose matrix in the standard basis $\E$ is the transpose of the matrix of $C_\B$. We define $M_\B = C_\B C_\B^T \in L(M_n)$.
\end{definition}

Since a linear map is uniquely determined by its values on a basis, we see that a change of basis map is uniquely determined up to re-orderings of the basis elements. Thus, in the setting of $M_n$ there will be $(n^2)!$ change of basis maps.  However, the map $M_\B$ is independent of the choice of change of basis map or matrix.  Indeed, fix one change of basis map $C_\B$. Then every change of basis map has the form $C_\B\circ P$, where $P \in M_n$ is a linear map which permutes the basis $\E$. Since $P$ sends the orthonormal basis $\E$ to itself, its matrix in that basis is orthogonal.  Thus $M_\B$ is unchanged if we replace  $C_\B$ by $C_\B P$.   

Fortunately, we will find that the results that we seek depend on $M_\B$ and are independent of the choice of change of basis map $C_\B$. In particular,  our conditions determining when the duality map $\D_\B$ is an order isomorphism or a complete order isomorphism will be expressed in terms of the map $M_\B$.

\begin{theorem} \label{main} If $\B$ is a basis of $M_n$, then   the duality map is given by
  $\D_\B= \D_\E \circ M_\B^{-1}$.
\end{theorem}

\begin{proof} Let $\B = \{X_1, \ldots, X_{n^2}\}$, and let  $\{\widehat X_1, \ldots, \widehat X_{n^2}\}\subset M_n^d$ be the dual basis. We write $\E = \{E_1, \ldots, E_{n^2}\}$.

 Define $Y_j = \D_\E^{-1}(\widehat X_j)= \Gamma_\E(\widehat X_j) \in M_n$. We are going to show
 \begin{equation}\label{eqkey} C_\B^T Y_j = E_j \text{ for all $j$}.
 \end{equation}   For each $i, j$, using \eqref{eq:density matrix}
 \begin{align}
 \<E_i, C_\B  ^*\overline{Y}_j \> &= \<C_\B  E_i, \overline{Y}_j\> = \<X_i, \overline{Y}_j\> 
 = \tr(X_i \overline{Y}_j^*\> = \tr (XY_j^t)\cr
  &= \tr(X_i (\Gamma_E(\widehat X_j))^t) 
  = \widehat X_j(X_i) = \delta_{ij}.
 \end{align}
 It follows that $C_\B  ^*\overline{Y}_j = E_j$. Taking conjugates gives \eqref{eqkey}.
 
 Now $$C_\B  C_\B^T D_\E^{-1}\widehat X_j =  C_\B  C_\B^T Y_j = C_\B  E_j = X_j = D_\B^{-1}\widehat X_j,$$ so by linearity $M_\B D_\E^{-1} = D_\B^{-1}$. Thus $D_\B = D_\E M_\B^{-1}$.

 %$\D_\E(Y_j) = \widehat X_j$ 
% \eqref{eq:density matrix} implies $\tr((Y_j)^t X_{i}) = {\widehat X}_j(X_{i}) = \delta_{ij}$. Thus for all $i,j$, using $(A^t)^* = \overline{A}$ for $A \in M_n$ we have
%$$\delta_{ij}  = \tr(Y_j^t X_i)  =\tr(X_iY_j^t)= \< \cC_\B(E_i), (Y_j^t)^*\> = \<E_i, \cC_\B^*(\overline{ Y}_j)\>.  $$
%Thus
% $$\<E_i, \cC_\B^*(\overline {Y}_j)\> = \<E_i, E_j\>$$
% for all $i, j$, which implies $\cC_\B^*(\overline {Y}_j) = E_j$ for all $j$. Hence
% $\overline{Y}_j = (\cC_\B^*)^{-1}E_j$, so $Y_j = \overline{(\cC_\B^*)^{-1}E_j} $. Since $\overline{\cC_\B^*} = \cC_\B^T$ and $\overline{E_j} = E_j$ then $Y_j = (\cC_\B^T)^{-1}(E_j)$, so
% $$\wD_\B(X_j) =\Gamma_\E(\D_\B  (X_j)) = \Gamma_\E(\widehat X_j)   = Y_j = (\cC_\B^T)^{-1}(E_j) = (\cC_\B^T)^{-1}(\cC_\B^{-1}(X_j)).$$
% By linearity, $\wD_\B(X) = (\cC_\B\cC_\B^T)^{-1}(X)$ for all $X \in M_n$.
\end{proof}

\begin{notation} If $C \in M_n$, then $\Phi_C:M_n \to M_n$ is the completely positive map defined by $\Phi_C(X) = CXC^*$.
\end{notation}

Recall that a map $\Psi:M_n \to M_p$ is called {\em completely co-positive} if and only if $t \circ \Psi$ is completely positive, and a map $\Psi:M_n \to M_n$ is called a {\em co-positive order isomorphism} provided that its composition $t \circ \Psi$ with the transpose map $t$ on $M_n$ is a complete order isomorphism.   Here the order of composition with the transpose map doesn't matter as shown by the next result.

\begin{proposition} Let $\Phi:M_n \to M_p$ then $t \circ \Phi$ is completely positive if and only if $\Phi \circ t$ is completely positive.
\end{proposition}
\begin{proof}
By Choi's result \cite{Choi}, $t \circ \Phi$ is completely positive if and only if
$\big( \Phi(E_{i,j})^t \big)$ is positive, while $\Phi \circ t$ is completely positive if and only if $\big( \Phi(E_{j,i}) \big)$ is positive. But these $np \times np$ block matrices are transposes of each other.
\end{proof}

We start with the following description of order automorphisms of $M_n$ and their partition into completely positive and completely co-positive maps.  It is a consequence of more general results of Kadison \cite{Kadison} relating isometries, Jordan isomorphisms, order isomorphisms, and *-isomorphisms of C*-algebras, specialized to  $M_n$ viewed as a C*-algebra.

\begin{lemma}\label{orderiso} Let $\Phi:M_n \to M_n$ be an order isomorphism.  Then there exists an invertible $C \in M_n$ such that either $\Phi = \Phi_C$ or $\Phi = t \circ \Phi_C$.  In the first case, $\Phi$ is a complete order isomorphism, and in the second case $\Phi$ is a co-positive order isomorphism.  If $n>1$, both cases cannot occur simultaneously. 
\end{lemma}

\begin{proof} First assume $\Phi$ is unital, i.e., that $\Phi(I) =  I$. For Hermitian matrices $A$, we have $\|A\|= \sup \{\lambda \in \R \mid -\lambda I \le A \le \lambda I\}$, so a unital order isomorphism is an isometry on Hermitian elements of $M_n$. It follows that $\Phi$ is an isometry on all of $M_n$ \cite[proof of Thm. 5]{Kadison}. Every unital isometry on $M_n$ is a Jordan isomorphism, i.e., preserves the Jordan product $A \circ B = (1/2)(AB + BA)$, cf. \cite[Thm. 7]{Kadison}.  Every Jordan isomorphism on $M_n$ is a *-isomorphism or *-anti-isomorphism \cite[Cor. 11]{Kadison}.  In the latter case, composing with the transpose map gives a *-isomorphism.  It is well known that every *-isomorphism of $M_n$ is conjugation by a unitary, see, for example,  \cite[Thm. 4.27]{AlfsenShultz}.  Thus there is a unitary $U$ such that $\Phi = \Phi_U$ or $\Phi = t \circ \Phi_U$.

Finally, let $\Phi$ be an arbitrary order isomorphism.   We will show $\Phi(I)$ is invertible.  Observe that $0 \le A \in M_n$ is invertible iff $A$ is an order unit, i.e., if for all $B = B^* \in M_n$ there exists  $\lambda \in \R$ such that $-\lambda A \le B \le \lambda A$. An order isomorphism takes order units to order units, so $\Phi(1)$ is invertible.  Let $D = \Phi(1)^{1/2}$, and define $\Psi = \Phi_{D^{-1}} \circ \Phi$. Then $\Psi$ is a unital order isomorphism, so by the first paragraph there exists a unitary $U$ such that $\Psi = \Phi_U$ or $\Psi = t \circ \Phi_U$. Then $\Phi = \Phi_D \circ \Psi =  \Phi_D \circ \Phi_U = \Phi_{DU}$, or else $\Phi = \Phi_D \circ t \circ \Phi_U = t \circ \Phi_C $ where $C = D^tU$.  Note that if both cases occur then $\Psi$ is both a *-isomorphism and a *-anti-isomorphism, which is possible only if $M_n$ is abelian, and that holds only when $n = 1$.
\end{proof}

Note that Lemma \ref{orderiso} implies that a composition of two complete order isomorphisms, or two co-positive order isomorphisms, is a complete order isomorphism, and a composition of a complete order isomorphism and a co-positive order isomorphism (in either order) is a co-positive order isomorphism.

\begin{theorem}\label{thm1} Let $\B$ be a basis of $M_n$.  Then $\D_\B$ is an order isomorphism iff there exists $C \in M_n$ such that either (1) $M_\B = \Phi_C$ or (2) $M_\B = t \circ \Phi_C$.  In the first case, $\D_\B$ is a complete order isomorphism, and in the second case it is a co-positive order isomorphism.   
\end{theorem}

\begin{proof} By Theorem \ref{main}, we have $\D_\B= \D_\E \circ M_\B^{-1}$.  Since $\D_E$ is a complete order isomorphism (Theorem  \ref{PTT}), then $D_\B$  is a complete order isomorphism (respectively, co-positive order isomorphism) if and only if $M_\B$ is a complete order isomorphism (respectively, co-positive order isomorphism). Now the theorem follows from Lemma \ref{orderiso}.
\end{proof}

\begin{corollary} Let $\B= \{B_j: 1 \le j \le n^2 \}$ be a basis for $M_n$. 
%and $\cC_\B:M_n \to M_n$ a change of basis. Then:
\begin{enumerate}
\item $M_\B= \Phi_C$ for some $C \in M_n$ if and only if $\Gamma_\B:M_n^d \to M_n$ is a complete order isomorphism,
\item  $M_\B^T = t \circ \Phi_C$ for some $C \in M_n$ if and only if $t \circ \Gamma_\B:M_n^d \to M_n$ is a complete order isomorphism.
\end{enumerate}
\end{corollary} 
\begin{proof} We prove the second statement.  By Theorem~\ref{thm1}, $M_\B =t \circ \Phi_C$ iff $D_\B$ is a co-positive order isomorphism. This is equivalent to $\Gamma_\B  = D_\B^{-1}$ being a co-positive order isomorphism, and hence to $t \circ \Gamma_\B $ being a complete order isomorphism.

% We prove the second statement. We have that $\cC_\B \cC_\B^T = t \circ \Phi_C$ if and only if $\wD_\B$ is a co-positive order isomorphism which by the last propostion is if and only if $\wD_\B \circ t = \Gamma_\E \circ \D_\B \circ t$ is a complete order isomorphism. But this is if and only if $t^{-1} \circ \D_\B^{-1} \circ \Gamma_\E = t \circ \Gamma_\B \circ \D_\B$ is a complete order isomorphism. Since $\D_\B$ is a complete order isomorphism, this is equivalent to $t \circ \Gamma_\B$ being a complete order isomorphism.
 
\end{proof}

Both Choi and Jamio\l kowski have defined useful correspondences that associate a matrix in $M_n \otimes M_p$ with each linear map $\Phi:M_n \to M_p$. Choi's correspondence is $\Phi \mapsto C_\Phi$, where
\begin{equation}
\label{Choimatrix}C_\Phi = \sum_{ij} E_{ij} \otimes \Phi(E_{ij}).
\end{equation}
As remarked in the proof of Theorem \ref{PTT}, positivity of the Choi matrix \eqref{Choimatrix} is equivalent to positivity of $\sum_{ij} \Phi(E_{ij}) \otimes E_{ij}$, and it is this latter form that we generalize below.  
We now describe a related correspondence  defined by Jamio\l kowski \cite{Jamiol}.  If $\Phi:M_n \to M_p$, then $\J(\Phi)$ is defined by the condition $\<\J(\Phi), A^* \otimes B\> = \<\Phi(A), B\>$ for all $A \in M_n$, $B \in M_p$. 
This is equivalent to 
\begin{equation}\label{Jam} J(\Phi) = \sum_{ij} E_{ij}^* \otimes \Phi(E_{ij}).
\end{equation}
Regarding our current investigation, the Choi matrix $C_\Phi$ has the property that $C_\Phi \ge 0$ iff $\Phi$ is completely positive, \cite{Choi}. The Jamio\l kowski correspondence has the property that in \eqref{Jam}, $\J(\Phi)$ is unchanged if the basis $\{E_{ij}\}$ is replaced by any orthonormal basis of $M_n$. Our correspondence will be closer to Choi's.

\begin{corollary}\label{genlchoi} Let $\B= \{B_j: 1 \le j \le n^2 \}$ be a basis for $M_n,$ and let $\Psi:M_n \to M_p$ be a linear map.
\begin{enumerate}
\item If $M_\B= \Phi_C$ for some $C \in M_n,$ then $\Psi$ is completely positive if and only if $\sum_{j=1}^{n^2} \Psi(B_j) \otimes B_j \in (M_p \otimes M_n)^+.$ 
\item  If $M_\B = t \circ \Phi_C$ for some $C \in M_n,$ then $\Psi$ is completely positive if and only if $\sum_{j=1}^{n^2} \Psi(B_j) \otimes B_j^t \in (M_p \otimes M_n)^+.$
\item If $M_\B = t \circ \Phi_C$ for some $C \in M_n,$ then $\Psi$ is completely co-positive if and only if $\sum_{j=1}^{n^2} \Psi(B_j) \otimes B_j \in (M_p \otimes M_n)^+.$
\end{enumerate}
\end{corollary}
\begin{proof} To prove the first statement, for $1 \le k, l \le p$ define $f_{k,l} \in M_n^d$ by  $\Psi(X) = (f_{k,l}(X))$.  Then by the definition of the order on $M_p(M_n^d)$ discussed earlier, $\Psi$ is completely positive if and only if $(f_{k,l}) \in M_p(M_n^d)^+$ which holds if and only if $(\Gamma_\B(f_{k,l})) \in M_p(M_n)^+.$ But as in the proof of Theorem~1, we have that
\[ (\Gamma_\B(f_{k,l})) = \sum_{j=1}^{n^2} \Psi(B_j) \otimes B_j.\]

To prove the second statement, note that $\Psi$ is completely positive if and only if $(t \circ \Gamma_\B(f_{k,l})) \in M_p(M_n)^+$ and this matrix is seen to be equal to
\[ \sum_{j=1}^{n^2} \Psi(B_j) \otimes B_j^t.\]

For the third statement, replace $\Psi$ by $t \circ \Psi$ in the second statement. This shows $\Psi$ is completely co-positive iff $\sum_{j=1}^{n^2} \Psi(B_j)^t \otimes B_j^t \in (M_p \otimes M_n)^+$, and now applying the transpose map gives (3).
\end{proof}

We now point out that bases with the properties indicated in the Corollary are related to entanglement witnesses. Indeed suppose $\B = \{B_1, \ldots, B_{n^2}\}$ is a basis of $M_n$ for which Corollary \ref{genlchoi} (1) holds. Taking $\Phi = I$ we have $\sum_i B_i \otimes B_i \ge 0$.  Let $\Phi:M_n \to M_n$ be a map that is positive but not completely positive. Define
$$B_0 = \sum_i B_i \otimes B_i \text{ and } B_\Phi = \sum_i \Phi(B_i) \otimes B_i .$$
Then for any positive $X, Y$, since $\Phi \ge 0$, we have
\begin{align}
\<B_\Phi, X \otimes Y\> &= \<\sum_i \Phi(B_i) \otimes B_i , X \otimes Y\> \cr
&= \<(\Phi \otimes I)B_0, X\otimes Y) = \<B_0, \Phi^*(X) \otimes Y\> \ge 0,
\end{align}
and hence $B_\Phi$ is $\ge 0$ on all separable states.  Since $\Phi$ is not completely positive, then $B_\Phi \not\ge 0$, so there is a state $A$ such that $\<B_\Phi , A\> \not\ge 0$. Such a state is then entangled, so $B_\Phi$ is an entanglement witness.

\section*{Examples}

\begin{notation}  If $x, y \in \C^n$, then $R_{x,y} \in M_n$ is the rank one operator defined by
$R_{x,y}z = \<z, y\>x$.  Observe that $E_{ij} = R_{e_i, e_j}$.
\end{notation}

\begin{proposition} \label{rank one}  Let $(\lambda_{ij}) \in M_n$, with all $\lambda_{ij}$ nonzero, and let $\B$ be the basis $\{\lambda_{ij}E_{ij}\}$.
Then $\D_\B$ is an  order isomorphism if and only if the matrix $(\lambda_{ij}^2)$ is positive semi-definite with rank one. In that case,  there are scalars $\alpha_1, \ldots, \alpha_n$ such that $\lambda_{ij}^2 = \alpha_i\overline{\alpha}_j $, and if $C= \diag(\alpha_1, \ldots, \alpha_n)$, then ${\wD}_\B= \Phi_C$, and hence $\D_\B$ is a complete order isomorphism.
\end{proposition}

\begin{proof} Note that $\cC_\B  $ is diagonal for the standard basis of $M_n$, so $\cC_\B  ^T = \cC_\B  $. Thus $M_\B(E_{ij})= (\cC_\B \cC_\B^T)( E_{ij} )= \lambda_{ij}^2 E_{ij}$.

Suppose that the map $\wD_\B:M_n \to M_n$ is an order isomorphism. We first consider the case where $\wD_\B^{-1} =M_\B = \Phi_C$.  Then $\Phi_C(E_{ij}) = \lambda_{ij}^2 E_{ij}$, so
$$\lambda_{ij}^2 E_{ij} = CE_{ij}C^* = CR_{e_i,e_j}C^* = R_{Ce_i, Ce_j}.$$
The ranges of the two sides must coincide, so for each $i$ there is a scalar $\alpha_i$ such that  $Ce_i = \alpha_i e_i$. Substituting into the displayed equation gives $\lambda_{ij}^2 E_{ij} = \alpha_i\overline{\alpha_j} E_{ij}$, so $\lambda_{ij}^2 = \alpha_i\overline{\alpha_j}$.  Thus the matrix $(\lambda_{ij}^2)$ has rank one and is positive. Conversely, if $(\lambda_{ij}^2)$ is positive with rank one, then there are nonzero scalars $\alpha_1, \ldots, \alpha_n$ such that $\alpha_i\overline{\alpha_j} = \lambda_{ij}^2$.  If $C= \diag(\alpha_1, \ldots, \alpha_n)$, then one readily verifies that $M_\B= \Phi_C$.  Then $\wD_\B^{-1}$ is a complete order isomorphism, and hence so is $\D_\B$.

Now we examine the possibility that $\cC_\B \cC_\B^T = t \circ \Phi_C$. Then
$$\lambda_{ij}^2 E_{ij} = (CE_{ij}C^*)^t = {C^*}^t E_{ji} C^t.$$
Let $D = {C^*}^t$, so that $\lambda_{ij}^2 E_{ij} = DE_{ji}D^*$. Then for all $i, j$
$$\lambda_{ij}^2 R_{e_i,e_j} = R_{De_j, De_i}.$$
This implies that $De_j$ is a multiple of $e_i$ for all $i, j$, which is impossible. 
\end{proof}

\begin{example}\label{onbasis} If $C \in M_n$ is invertible, then $\Phi_C^T = \Phi_{C^t}$, so $\Phi_C\Phi_C^T= \Phi_{CC^t}$. Hence  for the basis 
$ \B = \{\Phi_C(E_{ij})\}$, we have $M_\B =  \Phi_{CC^t}$, so by Theorem~\ref{thm1}, $\B$ has the property that the map from this basis  to its dual basis is a complete order isomorphism.  In particular, if $\{F_{ij}\}$ is any system of matrix units for $M_n$, there is a unitary $V$ such that $\Phi_V$ satisfies $\Phi_V(E_{ij})  = F_{ij}$ for all $i, j$, and so the map from $\{F_{ij}\}$ to its dual basis is a complete order isomorphism.

On the other hand, if $U: M_n \to M_n$ is unitary with respect to the Hilbert Schmidt inner product and takes $\E$ to a basis $\B$, it need not be the case that the duality map $\D_\B$ is a complete order isomorphism, as can be seen from Proposition~\ref{rank one} with $\lambda_{11} = i$ and $\lambda_{ij} = 1$ for $(i,j) \not= (1,1)$. Hence, not every orthonormal basis of $M_n$ has the property that the duality map is an order isomorphism.\end{example}

%\begin{example} The transpose map on $M_n$ has a symmetric matrix with respect to the standard basis, so $\wP_t = tt^t = tt = I$. Thus  $\Phi_t$ is a complete order isomorphism.  However, since the transpose map simply permutes the members of the standard basis $E_{ij}$, this is evident anyway from Lemma \ref{PTT}.
%\end{example}

%\begin{example} Of course, the map that takes a basis to its dual basis is independent of the order in which the basis elements are listed, and hence $\D_\B$ and $\Phi_{T P}$ are the same. This is consistent with Theorem \ref{main}:
%$$\wP_{T P}^{-1} = (TP)(TP)^t = TPP^tT^t = \cC_\B \cC_\B^T = \wD_\B.$$
%\end{example}

For our next application we study the Pauli spin matrices. 

\begin{theorem} Let $\B = \{\sigma_0, \sigma_1, \sigma_2, \sigma_3\}$ be the Pauli spin matrices, i.e.,
$$\sigma_0 =  \begin{pmatrix}1&0\cr 0&1\end{pmatrix},\quad
\sigma_1 =  \begin{pmatrix}0&1\cr 1&0\end{pmatrix},\quad
\sigma_2 =  \begin{pmatrix}0&-i\cr i&0\end{pmatrix},\quad
\sigma_3 =  \begin{pmatrix}1&0\cr 0&-1\end{pmatrix}.
$$
Then the duality map $\D_\B$ is a co-positive order isomorphism.  Furthermore,  let $\Psi: M_{2^n} \to M_p$ be a linear map. Then $\Psi$ is completely positive if and only if
\[ \sum_{i_1,...,i_n=0}^3 \Psi(\sigma_{i_1} \otimes \cdots \otimes \sigma_{i_n}) \otimes \sigma_{i_1}^t \otimes \cdots \otimes \sigma_{i_n}^t \]
is a positive $2^np \times 2^np$ matrix.

Similarly, $\Psi$ is completely co-positive if and only if
\[ \sum_{i_1,...,i_n=0}^3 \Psi(\sigma_{i_1} \otimes \cdots \otimes \sigma_{i_n}) \otimes \sigma_{i_1} \otimes \cdots \otimes \sigma_{i_n} \]
is a positive $2^np \times 2^np$ matrix.
\end{theorem}
\begin{proof}
Let $\cC_\B$ be the linear map such that
$$\cC_\B(E_{11}) = \sigma_0, \quad \cC_\B(E_{12}) = \sigma_1,\quad \cC_\B(E_{21})  = \sigma_2,\quad \cC_\B(E_{22})= \sigma_3.$$
Then the matrix for $M_\B =\cC_\B \cC_\B^T$ in the standard basis of $M_2$ is
$$[\cC_\B][\cC_\B^T] = \begin{pmatrix} 2&0&0&0\cr 0&0&2&0\cr 0&2&0&0\cr 0&0&0&2\end{pmatrix}$$
which  is twice the matrix of the transpose map $t:M_2 \to M_2$. Thus $\D_\B$ in this case is a co-positive order isomorphism.

Applying Corollary~\ref{genlchoi}, we see that a map $\Psi:M_2 \to M_p$ is completely positive if and only if
\[ \sum_{j=0}^3 \Psi(\sigma_j) \otimes \sigma_j^t \in (M_p \otimes M_2)^+.\]
Using the explicit form of the Pauli matrices, we obtain that $\Psi$ is completely positive if and only if
\[ \begin{bmatrix} \Psi(\sigma_0) + \Psi(\sigma_3), & \Psi(\sigma_1)+ i \Psi(\sigma_2)\\ \Psi(\sigma_1) - i \Psi(\sigma_2), & \Psi(\sigma_0) - \psi(\sigma_3) \end{bmatrix} \]
is positive in $M_2(M_p),$ which is identical to Choi's theorem.

%Also note that $\sigma_j^t = \sigma_j$ for $j \ne 2,$ while $\sigma_2^t = - \sigma_2.$

On $M_{2^n}$ the tensored spin matrices 
$\B^{\otimes n} = \{\sigma_{i_1} \otimes \cdots \otimes \sigma_{i_n}\mid 0 \le i_j \le 3 \}$ form an orthonormal basis which we will call the spin basis. The standard basis of matrix units of $M_{2^n}$ consists of the tensor products of the matrix units of $M_2$. The map $\cC_{\B^{\otimes n}}: M_{2^n} \to M_{2^n}^d$ taking the standard basis of matrix units to this spin basis is then the tensor product of the maps on each factor $M_2$, so $(\cC_{\B^{\otimes n}})(\cC_{\B^{\otimes n}})^T$ will be the transpose map on $M_{2^n}$. Thus the map from the spin basis on $M_{2^n}$ to its dual basis will also be a co-positive order isomorphism.

Again applying Corollary~\ref{genlchoi} yields that a map, $\Psi: M_{2^n} \to M_p$ is completely positive if and only if
\[ \sum_{i_1,...,i_n=0}^3 \Psi(\sigma_{i_1} \otimes \cdots \otimes \sigma_{i_n}) \otimes \sigma_{i_1}^t \otimes \cdots \otimes \sigma_{i_n}^t \]
is a positive $2^np \times 2^np$ matrix.

Similarly, $\Psi$ is completely co-positive if and only if
\[ \sum_{i_1,...,i_n=0}^3 \Psi(\sigma_{i_1} \otimes \cdots \otimes \sigma_{i_n}) \otimes \sigma_{i_1} \otimes \cdots \otimes \sigma_{i_n} \]
is a positive $2^np \times 2^np$ matrix.

\end{proof}

For our final application we study the map from the Weyl basis to
its dual basis. We will compute the duality map for the Weyl basis, with the conclusion that this map is a complete order isomorphism for $n =2$, but is not an order isomorphism for $n > 2$. Below $n> 1$ is a positive integer, and all indices are viewed as members of $\Z_n$.   

\begin{definition} Let $e_0, \ldots, e_{n-1}$ be the standard basis of $\C^n$, and $\B=\{E_{ab} \mid a, b \in \Z_n\}$ the corresponding matrix units. Let $U, V \in M_n$ be defined by $Ve_j = z^j e_j$ and $Ue_j =e_{j+1}$ where $z = \exp(2\pi i/n)$  Then $\{{1\over\sqrt{n}}U^a V^b \mid a, b \in \Z_n\}$ is an orthonormal basis for $M_n$ which we call the {\em Weyl basis $\W$.} 
\end{definition}

The unitary matrices $ \{U^a V^b \mid a, b \in \Z_n\}$ are usually called the \emph{discrete Weyl matrices} or the \emph{generalized Pauli matrices}.

\begin{lemma} \label{L1} Define $C_\W \in L(M_n)$ by $C_\W(E_{ab}) = {1\over \sqrt{n}}U^aV^b$. With respect to the standard basis of matrix units, we have the following matrix entries for $C_\W$ and $C_\W C_\W^T $:
$$[C_\W ]_{ab,cd} = z^{db} \delta_{b+c,a} \text{ and } [C_\W C_\W ^T]_{ab,cd} =  \delta_{b,-d}\delta_{a,c-2d}.$$
\end{lemma}

\begin{proof} We have $U^aV^b e_j =  z^{bj} E_{j+a,j} e_j$  so
$$U^aV^b = \sum_j z^{bj} E_{j+a,j}.$$
Thus
\begin{align}\label{Tab}
[C_\W ]_{ab,cd} &= \< C_\W (E_{cd}),E_{ab}\> \cr
&= {1\over \sqrt{n}}\<U^cV^d, E_{ab}\> \cr
&= {1\over \sqrt{n}}\<\sum_j z^{dj}E_{j+c,j}, E_{ab}\> \cr
&= {1\over \sqrt{n}}\sum_j z^{dj} \delta_{j+c,a}\delta_{j,b} \cr
&={1\over \sqrt{n}}z^{db} \delta_{b+c,a}
\end{align}
Now
$$(C_\W C_\W ^T)_{ab, cd} = \sum_{jk} [C_\W ]_{ab, jk}[C_\W ^T]_{jk, cd}=\sum_{jk} [C_\W ]_{ab, jk}[C_\W ]_{cd, jk}.$$
In the first factor $[C_\W ]_{ab, jk}$ of the last sum we use the expression \eqref{Tab} for $[C_\W ]_{ab,cd}$ with the substitutions $c\to j$ and $d\to k$.  In the second factor $[C_\W]_{cd, jk}$ we use \eqref{Tab} with the substitutions $a\to c, b\to d, c\to j, d\to k$. We get
$$(C_\W C_\W^T)_{ab, cd}= {1\over n}\sum_{jk} z^{kb} \delta_{b+j,a}z^{kd} \delta_{d+j,c}.$$
The summands will be nonzero if and only if $j = a-b = c-d$ (mod $n$). Thus
$$(C_\W C_\W^T)_{ab, cd}= {1\over n}\delta_{a-b,c-d}\sum_{k}  z^{k(b+d)}.$$
The sum will be zero unless $b+d = 0$, in which case it has the value $n$. Thus
$$(C_\W C_\W^T)_{ab, cd}= \delta_{a-b,c-d}\delta_{b+d,0} =\delta_{b,-d}\delta_{a,c-2d}.$$
\end{proof}

Note that Lemma~\ref{L1} gives 
\begin{equation} \label{key}
(C_\W C_\W^T)(E_{c,d}) = E_{c-2d,-d}\, ,
\end{equation}
 so in particular $C_\W C_\W^T$ acts as a permutation on the basis of matrix units.

\begin{corollary}  For the Weyl basis $\W$,  the duality map $\D_\W$ is a complete order isomorphism if $n = 2$, and is not an order isomorphism for $n > 2$.
\end{corollary}

\begin{proof} If $n = 2$, then from \eqref{key} $C_\W C_\W^T$ is the identity map, and hence $\D_\W$ is a complete order isomorphism.

Now let $n > 2.$ Suppose first (to reach a contradiction) that $C_\W C_\W^T = \Phi_C$ for some invertible $C \in L(M_n)$. Then by \eqref{key},
$$E_{-d,-d} = (C_\W C_\W^T)(E_{dd})  = \Phi_C E_{dd} = R_{Ce_d, Ce_d}.$$
Thus for all $d$ there are scalars $\lambda_d$ of modulus one such that $Ce_d = \lambda_d e_{-d}$. Then 
$$E_{c-2d,-d} = (C_\W C_\W^T)(E_{cd}) =\Phi_C(E_{cd}) = R_{Ce_c, Ce_d} = \lambda_c\overline{\lambda}_d E_{-c,-d}.$$ This implies $c-2d = -c$, so $2c = 2d \bmod{n}$ for all $c, d$.   This is impossible for $n > 2$.

Now suppose $C_\W C_\W^T = t \circ \Phi_C$. Then again applying \eqref{key},
$$E_{-d,-d} = (C_\W C_\W^T)(E_{d,d}) = (\Phi_C(E_{d,d}))^t=  (R_{Ce_d, Ce_d})^t = (R_{\overline{Ce_d}, \overline{Ce_d}}).$$
This implies that for all $d$, $\overline{Ce_d}$ is a multiple of $e_{-d}$, and hence $Ce_d$ is a multiple of $\overline{e_{-d}} = e_{-d}$. As above
$$E_{c-2d,-d}=(C_\W C_\W^T)(E_{cd}) = (R_{Ce_c, Ce_d})^t = R_{\overline{Ce_d}, \overline{Ce_c}} = \lambda_c\overline{\lambda}_d E_{-c,-d}.$$
This implies $c-2d = -c$ for all $c, d$, which again is impossible for $n > 2$.
\end{proof}

\begin{remark} The Weyl basis for $n = 2$ is slightly different than the Pauli spin basis. Indeed, one has
$${1\over \sqrt{2}}I = {1\over \sqrt{2}}\begin{pmatrix}1&0\cr 0&1\end{pmatrix}, \quad
{1\over \sqrt{2}}U = {1\over \sqrt{2}}\begin{pmatrix}0&1\cr 1&0\end{pmatrix},$$
$$ 
{1\over \sqrt{2}}V = {1\over \sqrt{2}}\begin{pmatrix}1&0\cr 0&-1\end{pmatrix}, \quad
{1\over \sqrt{2}}UV = {1\over \sqrt{2}}\begin{pmatrix}0&-1\cr 1&0\end{pmatrix}$$
These are the Pauli spin matrices except for normalization and a missing factor of $i$ in the last. Again, applying Corollary~\ref{genlchoi} yields the usual Choi condition.
   
Recall that for the Pauli spin matrices we found $C_\B C_\B^T$ was the transpose map, so $\D_\B$ in that case was co-positive.
\end{remark}

\begin{remark} Note that for $M_{2^n}$ if we take tensors of the
  Weyl basis for $M_2 $, then we will get another basis for which the duality map is a complete order isomorphism.
\end{remark}

\section*{The Conjugate Linear Duality Map}

Using the fact that $M_n$ is a Hilbert space, we also have a canonical conjugate linear isomorphism between $M_n$ and its dual space. This map is unaffected by whether we make the inner product conjugate linear in the first or second variable, so we use the physics convention that inner products are conjugate linear in the first variable. Thus, the inner product on $M_n$ can be given by
\[ \langle A,B \rangle = Tr(A^*B)\]
and the conjugate linear Hilbert space duality map is given by
\[ \D_d:M_n \to M_n^d  \text{ where } \D_d(A)(B) = Tr(A^*B).\]
The inverse of this map
\[ \Gamma_d = \D_d^{-1}: M_n^d \to M_n\]
sends the linear functional $f_A(B) = Tr(A^*B)$ to the matrix $A$ which is the adjoint of the density matrix. 

\begin{proposition} The duality maps $\D_d$ and $\Gamma_d$ are conjugate linear complete order isomorphisms.
\end{proposition}
\begin{proof} By Corollary~\ref{densitymatrixcor} the linear map that sends a functional to its density matrix is a co-positive order isomorphism. Hence, the map that sends a functional to the transpose of its density matrix is a complete order isomorphism.  But a matrix $( c_{i,j})$ is positive if and only if the matrix $(\overline{c_{i,j}})$ is positive. Thus, the duality map $\Gamma_d$, which sends a functional to the adjoint of its density matrix is a complete order isomorphism. Consequently, so is its inverse $\D_d.$
\end{proof}

The above result has a nice interpretation in terms of bases.  Choi's characterization says that a map $\Phi:M_n \to M_p$ is completely positive iff the matrix $C_\Phi$ defined in \eqref{Choimatrix} is positive.  As observed in Example \ref{onbasis}, in the definition of $C_\Phi$, the basis $\{E_{ij}\}$ can't  be replaced by an arbitrary orthonormal basis.  The following result provides an alternate description of the Choi matrix that does have this independence property.
Given a matrix $B= (b_{i,j})$ we set $\overline{B} = ( \overline{b_{i,j}}).$ 

\begin{proposition}\label{prop22} Let $\{ B_l \}_{l=1}^{n^2}$ be an orthonormal basis for $M_n.$ A complex linear map $\Phi: M_n \to M_p$ is completely positive if and only if
\begin{equation}\label{choi2}
\sum_{l=1}^{n^2} \overline{B_l} \otimes \Phi(B_l)   
\end{equation}
is a positive $np \times np$ matrix. The matrix in \eqref{choi2} is  independent of the choice of orthonormal basis, and equals the  Choi matrix $C_\Phi$.
\end{proposition}

\begin{proof} Let $f_A \in M_n^d$ be given by $f_A(B) = Tr(A^*B),$ so that 
\[\Gamma_d(f_A) = A= \sum_{l=1}^{n^2} \langle B_l, A \rangle B_l = \sum_{l=1}^{n^2} \overline{f_A(B_l)} B_l.\] Thus, with respect to this basis
\[ \Gamma_d(f) = \sum_{l=1}^{n^2} \overline{f(B_l)} B_l \quad \text{for all $f \in M_n^d$}.\]

Let $(f_{ij}) \in   M_n^d$ be the matrix defined by $\Phi(A) = (f_{ij}(A))$ for $A \in M_n$. Recall that by definition, the matrix $(f_{ij})$ is positive iff $\Phi$ is completely positive. Using the fact that $\Gamma_d$ is a complete order isomorphism, we have that $\Phi$ is completely positive if and only if
\begin{equation}\label{eqz} (\Gamma_d(f_{i,j})) = \sum_{l=1}^{n^2} \big( \overline{f_{i,j}(B_l)} B_l \big) = \sum_{l=1}^{n^2} \overline{\Phi(B_l)} \otimes B_l 
\end{equation}
is a positive $np \times np$ matrix. Using that fact that a matrix is positive if and only if its complex conjugate matrix is positive, the proposition now follows by applying the *-isomorphism  that takes $A \otimes B$ to $B \otimes A$.

Finally, since  the matrix $(f_{ij}) \in M_n^d$ is determined by $\Phi$,  the matrix $\sum_l \overline{B_l} \otimes \Phi(B_l)$ is independent of the choice of orthonormal basis $\{B_l\}$. For the standard basis $\{E_{ij}\}$ the matrix \eqref{choi2} is just the Choi matrix, and hence the matrix in \eqref{choi2} equals the Choi matrix   for all orthonormal bases $\{B_l\}$.
\end{proof}

 Note that the matrix in \eqref{choi2} is the partial transpose of the matrix $J(\Phi)$ defined by Jamio\l kowski, cf. \eqref{Jam}.
%\subsection*{Conclusion}

%This finishes our investigation of the complete positivity of the map
%that takes a basis of $M_n$ to its dual basis. We have given a
%necessary and sufficient condition for this to be an order
%isomorphism, and to be completely positive or co-positive, and applied
%this to several examples. We have also generalized the Choi matrix by
%giving conditions under which the role of the standard basis
%$\{E_{ij}\}$ can be replaced by other bases.

\subsection*{Acknowledgements}
The research on this paper was begun during the program {\it Operator
  structures in quantum information theory} at the Banff International
Research Station.

\end{document}